\documentclass[letterpaper, 10 pt, conference]{ieeeconf}  
\usepackage{amsmath,amssymb,bm,mathtools}
\usepackage[noadjust]{cite}
\usepackage[capitalize]{cleveref}
\usepackage[norelsize,ruled,linesnumbered,noend]{algorithm2e}
\usepackage{xcolor}
\usepackage{fp, tikz, pgfplots}
\usepackage[footnotesize, skip=0pt]{caption}
\usepackage{accents,interval}
\intervalconfig{soft open fences,separator symbol =; }
\usetikzlibrary{arrows}
\usetikzlibrary{shapes}
\usetikzlibrary{patterns}
\usetikzlibrary{decorations.pathmorphing}
\usetikzlibrary{decorations.pathreplacing}
\usetikzlibrary{decorations.shapes}
\usetikzlibrary{decorations.text}
\usetikzlibrary{shapes.misc}
\usetikzlibrary{decorations.markings}
\usetikzlibrary{arrows.meta}
\usetikzlibrary{backgrounds}
\usepgfplotslibrary{fillbetween}

\newtheorem{definition}{Definition}
\newtheorem{assumption}{Assumption}
\newtheorem{theorem}{Theorem}
\newtheorem{lemma}{Lemma}

\newtheorem{corollary}{Corollary}

\newcommand{\ubar}[1]{\underaccent{\bar}{#1}}

\makeatletter
\newcommand{\removelatexerror}{\let\@latex@error\@gobble}
\makeatother

\makeatletter
\newcommand{\subalign}[1]{%
	\vcenter{%
		\Let@ \restore@math@cr \default@tag
		\baselineskip\fontdimen10 \scriptfont\tw@
		\advance\baselineskip\fontdimen12 \scriptfont\tw@
		\lineskip\thr@@\fontdimen8 \scriptfont\thr@@
		\lineskiplimit\lineskip
		\ialign{\hfil$\m@th\scriptstyle##$&$\m@th\scriptstyle{}##$\crcr
			#1\crcr
		}%
	}
}
\makeatother

\pgfplotsset{colormap={blueyellow}{rgb(0cm)=(0,0,1); rgb(1cm)=(1,1,0)}}

\pgfplotsset{width=9\columnwidth /10, compat = 1.13, 
	height = 9\columnwidth /10, grid= major, 
	legend cell align = left, ticklabel style = {font=\scriptsize},
	every axis label/.append style={font=\small},
	legend style = {font=\scriptsize},title style={yshift=0pt, font = \small} }

\tikzset{cross/.style={cross out, draw=black, minimum size=10*(#1-\pgflinewidth), inner sep=0pt, outer sep=0pt},cross/.default={1pt}}

\tikzset{
	myarrow/.style={-{Triangle[length=2mm,width=2mm]}}
}

\allowdisplaybreaks
\IEEEoverridecommandlockouts                              

\title{\LARGE \bf
How Training Data Impacts Performance in Learning-based Control  
}

\author{Armin Lederer, Alexandre Capone, Jonas Umlauft, Sandra Hirche
\thanks{*This work was supported by the European Research Council (ERC)
Consolidator Grant ``Safe data-driven control for human-centric systems (CO-MAN)'' under grant agreement number 864686. 
A. L. gratefully acknowledges financial support from
the German Academic Scholarship Foundation.}
\thanks{All authors are with the Department of Electrical and Computer Engineering, Technical University of Munich, 80333 Munich, Germany
	{\tt\small [armin.lederer, alexandre.capone, jonas.umlauft, 
	hirche]@tum.de}}%
}

\begin{document}

\maketitle
\thispagestyle{empty}
\pagestyle{empty}

\begin{abstract}

When first principle models cannot be derived due to the complexity of the real 
system, data-driven methods allow us to build models from system observations. 
As these models are employed in learning-based control, the 
quality of the data plays a crucial role for the performance of the resulting 
control law. Nevertheless, there hardly exist measures for assessing training 
data sets, and the impact of the distribution of the data on the closed-loop 
system properties is largely unknown. This paper derives --- 
based on Gaussian process models --- an analytical relationship between the density of the 
training data and the control performance. We formulate a quality 
measure for the data set, which we refer to as $\rho$-gap, and derive the ultimate bound for the tracking 
error under consideration of the model uncertainty. We show how the $\rho$-gap 
can be applied to a feedback linearizing control law and 
provide numerical illustrations for our approach.\looseness=-1
\end{abstract}

\section{INTRODUCTION}
Model-based control requires an accurate mathematical description of the plant that is to be controlled. Classical system identification methods postulate parametric models using prior assumptions, and tune their parameters based on observations to achieve high model accuracy~\cite{Ljung1998}. However, these methods are prone to yielding poor models if a wrong parametric structure is assumed, i.e., if insufficient prior knowledge on the system structure is available. This is often the case for highly complex systems, e.g., settings where humans are part of the control loop. In order to overcome these shortcomings, learning-based control employs non-parametric data-driven models, which only require little prior system knowledge in comparison to classical parametric models~\cite{Kocijan2004}. Such modeling techniques 
strongly rely on (potentially noisy) observations, making a formal analysis of 
the resulting control performance difficult. Hence, these techniques remain too unreliable for safety-critical applications~\cite{Fisac2019}. \looseness=-1

To overcome this drawback, recent research has focused on the stability of 
learning-based control approaches using Gaussian process (GP) 
models~\cite{Umlauft2017a}. GPs can capture model uncertainty, which allows us to derive probabilistic model 
error bounds~\cite{Lederer2019}. The applications of GP-based methods range from safe 
controller optimization for quadrotors~\cite{Berkenkamp2016} to computed torque 
control in robotics~\cite{Beckers2019} and feedback linearization 
for aircraft systems~\cite{Chowdhary2015}.\looseness=-1

Despite the widespread use of GPs in control, there exist only a few tools to assess 
the quality of the training set. So far, most measures used to quantify data 
quality have been based on information-theoretical measures, e.g., 
information gain~\cite{Koller2018,Alpcan2015}. These techniques 
assess data in global terms, without taking into account locally varying 
requirements on the data due to the control structure and the task. 
Since the relationship between data distributions and control performance is largely 
unknown, random sampling-based approaches have recently been employed to estimate the effect of 
data on learning-based control systems \cite{Capone2020}. However, sampling-based approaches 
are computationally expensive and provide no direct insight into the interrelation between 
training data and control error. Therefore, deriving an analytical measure is crucial to 
improve our understanding of this relationship, and is essential to enhance the efficiency of exploration in
active learning, training data selection to implement machine learning with limited computational budget, 
and cautious control design.\looseness=-1

The main contribution of this paper is a novel measure, called $\rho$-gap, 
to assess training data sets from a control theoretical perspective. Based on 
the model uncertainty of a GP model, we investigate the uncertainty-dependent Lyapunov 
stability conditions for a control-affine closed-loop system. This analysis allows 
insights on how data should be collected, which is becoming particularly useful in 
exploration tasks where high data-efficiency is required. As an example, we 
derive a novel uncertainty-dependent ultimate bound of the tracking error 
for a feedback linearizing control law and show how the density of the training
data affects this bound.\looseness=-1

The paper is structured as follows: \Cref{sec:problem} defines the problem 
setting, after which GP regression and the required model 
error bounds are introduced in  \cref{sec:GPR}. 
The control law is presented in \cref{sec:Control}, including the derivation of 
the ultimate bound and the proposed quality measure of the data. The results 
are numerically illustrated in \cref{sec:simulation}, followed by the 
conclusion in \cref{sec:conclusion}.

\section{PROBLEM STATEMENT}
\label{sec:problem}
We consider a single-input system in the canonical form\footnote{Notation: Lower/upper case bold
	symbols denote vectors/matrices,~$\mathbb{R}_{+,0}$/$\mathbb{R}_+$ all real positive numbers 
	with/without zero, respectively.~$\bm{I}_n$ denotes the~$n\times n$ identity matrix and 
	$\|\cdot\|$ the Euclidean norm.\looseness=-1}
\begin{align}
\label{eq:sys}
	\dot{x}_1=x_2, \quad \dot{x}_2=x_3, \quad \ldots \quad 
	\dot{x}_{d_x}=f(\bm{x})+g(\bm{x})u
\end{align}
with state~$\bm{x}=[x_1\quad \cdots\quad x_{d_x}]^T\in\mathbb{X}$ in a compact set $\mathbb{X}\subset\mathbb{R}^{d_x}$, 
input~\mbox{$u\in\mathbb{U}=\mathbb{R}$}, and unknown functions~$f\colon 
\mathbb{X}\to \mathbb{R}$ 
and~$g\colon 
\mathbb{X}\to \mathbb{R}$. 
Note that we restrict the following analysis to 
single-input systems due to notational convenience, but our results directly 
extend to multi-input control-affine systems. 
We assume that prior models~$\hat{f}\colon 
\mathbb{X}\to \mathbb{R}$ and~$\hat{g}\colon 
\mathbb{X}\to \mathbb{R}$ of the unknown functions
are given, and make the following 
assumptions on the unknown functions and the available data:\looseness=-1
\begin{assumption}
	\label{ass:data}
	A data set 
	\begin{align}
		\mathbb{D}=\left\{\bm{z}^{(n)}\coloneqq\begin{bmatrix}
		\bm{x}^{(n)}\\ u^{(n)}
		\end{bmatrix}, 
		y^{(n)}\!=\!\Delta\left(\bm{z}^{(n)}\right)\!+\!\omega^{(n)}\right\}_{n=1}^N
	\end{align}
	is available, which contains~$N$ pairs of
	noiseless measurements of the state~$\bm{x}^{(n)}$ and noisy measurements
	\begin{align}
	\label{eq:Delta}
		\Delta(\bm{z})=\Delta(\bm{x},u)\!=\!f(\bm{x})\!-\!\hat{f}(\bm{x})\!+\!(g(\bm{x})\!-\!\hat{g}(\bm{x}))u\!
	\end{align}\\\vspace{-0.8cm}\\
	perturbed by Gaussian 
	noise~$\omega^{(n)}\sim\mathcal{N}(0,s_{\mathrm{on}}^2)$, where~$\bm{z}\in 
	\mathbb{X}\times\mathbb{U}\subset \mathbb{R}^{d_z}$, and~$d_z = d_x 
	+1$.
\end{assumption}
\begin{assumption}
	\label{ass:prior}
	The unknown functions~$f(\cdot)$ and~$g(\cdot)$ admit Lipschitz  
	constants~$L_f$ and~$L_g$, respectively.\looseness=-1
\end{assumption}
\begin{assumption}
	\label{ass:fp}
	The sign of~$g(\cdot)$ is known and constant.
\end{assumption}

While Assumptions~\ref{ass:data} and~\ref{ass:prior} guarantee the existence of 
training data, and ensure that the unknown functions are well-behaved, \cref{ass:fp}  
guarantees global controllability, and thereby 
the existence of a stabilizing control law.
We consider the task of tracking a bounded reference trajectory\looseness=-1
\begin{align}
	\bm{x}_d=\begin{bmatrix}
	x_d&\dot{x}_d&\ldots&\frac{\mathrm{d}^{d_x-1}x_d}{\mathrm{d}t^{d_x-1}}
	\end{bmatrix}^T,
\end{align}
which can be seen as the generalization of set point 
regularization, to which our results carry over. We employ 
a control law $\pi:\mathbb{X}\rightarrow\mathbb{U}$, 
whose goal is to stabilize the tracking
error~$\bm{e}=\bm{x}-\bm{x}_d$ with the dynamics
\begin{align}
\label{eq:edyn}
\dot{e}_1=e_2,\quad \ldots \quad 
\dot{e}_{d_x}=f(\bm{x})+g(\bm{x})\pi(\bm{x})-\frac{\mathrm{d}^{d_x-1}x_d}{\mathrm{d}t^{d_x-1}}.
\end{align}
The control law $\pi(\cdot)$  
is based on a model learned by GP regression with a 
composite kernel
\begin{align}
\label{eq:compkern} 
	k(\bm{z},\bm{z}')=k_f(\bm{x},\bm{x}')+uk_g(\bm{x},\bm{x}')u',
\end{align}
and covariance functions~$k_f,k_g:\mathbb{R}^{d_x}\times\mathbb{R}^{d_x}\rightarrow\mathbb{R}$ 
as suggested in \cite{Umlauft2020}. Since this kernel reflects the structure of \eqref{eq:sys}, it allows us
to recover separate models for~$f(\cdot)$ and~$g(\cdot)$~\cite{Duvenaud2014}.
While other kernel functions
can also be used for a learning-based control approach \cite{Chowdhary2015}, 
they generally do not allow a separation of model components. Since this separation
is beneficial for the interpretability and intuitiveness of the relationship between training
data and control performance, we focus on composite kernels in the following.
The individual covariance functions
represent our prior knowledge about the unknown functions.
\begin{assumption}
	\label{ass:gp}
	Prior knowledge of the function is expressed through prior 
	GPs, i.e.,~$f(\cdot)\sim\mathcal{GP}(\hat{f}(\bm{x}),k_f(\bm{x},\bm{x}'))$ and
	$g(\cdot)\sim\mathcal{GP}(\hat{g}(\bm{x}),k_g(\bm{x},\bm{x}'))$.
\end{assumption}

This assumption imposes a probability distribution on the function space, which is shaped by the 
prior mean functions~$\hat{f}(\cdot),\hat{g}(\cdot)$ and the kernel functions 
$k_f(\cdot,\cdot),k_g(\cdot,\cdot)$. Thereby, it implicitly requires 
that the covariance kernels and prior mean functions are chosen suitably, i.e., 
$f(\cdot),g(\cdot)$ must be expressible in terms of those functions~\cite{Lederer2019}.
 
While the stability of control laws has been investigated 
under these assumptions \cite{Chowdhary2015, Umlauft2020}, the derived ultimate bounds do not depend on 
the training data. Hence, the impact of training data on the performance of the learning-based controller is 
unknown. We address this issue by developing  
a flexible measure of the quality of training data with respect to the control performance. In order to illustrate the flexibility
of the proposed quality measure, we derive a novel 
uncertainty-dependent ultimate tracking error bound
for feedback linearizing control of systems with
both $f(\cdot)$ and $g(\cdot)$ unknown, and apply
our quality measure to this problem.\looseness=-1

\section{GAUSSIAN PROCESS REGRESSION}
\label{sec:GPR}

A Gaussian process defines a  
distribution~$\mathcal{GP}(\mu_0(\bm{z}),k(\bm{z},\bm{z}'))$ over 
functions~$h:\mathbb{R}^{d_z}\rightarrow \mathbb{R}$ with prior 
mean~$\mu_0:\mathbb{R}^{d_z}\rightarrow\mathbb{R}$ and covariance~$k:\mathbb{R}^{d_z}\times\mathbb{R}^{d_z}
\rightarrow \mathbb{R}$, such that any finite number of evaluation points 
$\{\bm{z}_1,\ldots,\bm{z}_m\}$,~$m\in\mathbb{N}$ is assigned a Gaussian 
distribution~\cite{Rasmussen2006}. 
The prior mean function 
incorporates known parametric models into the regression, while the kernel~$k(\cdot,\cdot)$ 
encodes information about the structure of~$h(\cdot)$. Due to the structure of the  
modeling error~$h(\cdot)=\Delta(\cdot)$ in~\eqref{eq:Delta}, 
we employ a prior mean function 
$\mu_0(\bm{z})=\hat{f}(\bm{x})+\hat{g}(\bm{x})u$ and the composite kernel~$k(\cdot,\cdot)$
defined in \eqref{eq:compkern}. For the components of the composite kernel~$k(\cdot,\cdot)$,
we use squared exponential kernels~$k_f(\cdot,\cdot)$ and~$k_g(\cdot,\cdot)$,
defined as
\begin{align}
k_{f,g}(\bm{x},\bm{x}')=s_{f,g}^2\exp\left(-\frac{\|\bm{x}-\bm{x}'\|}{2l_{f,g}^2}\right),
\label{eq:SE}
\end{align}
where~$s_f^2,s_g^2\in\mathbb{R}_{+,0}$ and~$l_f,l_g\in\mathbb{R}_+$ denote the signal variances and 
length scales, respectively. 
Using the covariance function~$k(\cdot,\cdot)$, the 
elements of the data covariance matrix~$\bm{K}$ and the 
kernel vector~$\bm{k}(\bm{z})$ at a test point~$\bm{z}$ are given by~$K_{nn'}=k(\bm{z}^{(n)},\bm{z}^{(n')})$ 
and~$k_{n}(\bm{z})=k(\bm{z}^{(n)},\bm{z})$, respectively. Based on these definitions, 
the probability of~$h(\bm{z})$
conditioned on the training data~$\mathbb{D}$ as well as the test point~$\bm{z}$
is Gaussian with mean and variance
\begin{align}
\label{eq:mean}
\mu(\bm{z})&= \mu_0(\bm{z})+\bm{k}^T(\bm{z})\left(\bm{K}+s_{\mathrm{on}}^2\bm{I}_N\right)^{-1}\bm{y}\\
\sigma^2(\bm{z})&=k(\bm{z},\bm{z})-\bm{k}^T(\bm{z})(\bm{K}+s_{\mathrm{on}}^2\bm{I}_N)^{-1}\bm{k}(\bm{z}),
\end{align}
where the training outputs~$y^{(n)}$
are concatenated in the target vector~$\bm{y}=\begin{bmatrix}
y^{(1)}&\ldots&y^{(N)}
\end{bmatrix}^T$. Due to the definition of the composite kernel \eqref{eq:compkern}, we
can express the kernel vector as~$\bm{k}(\bm{z})=\bm{k}_f(\bm{x})+\bm{U}\bm{k}_g(\bm{x})u$,
where~$\bm{U}=\mathrm{diag}([u^{(1)}\ \ldots \ u^{(N)}])$ and the elements of~$\bm{k}_f(\cdot)$ 
and~$\bm{k}_g(\cdot)$ are defined as~$k_{f,n}(\bm{x})=k_f(\bm{x}^{(n)},\bm{x})$ 
and~$k_{g,n}(\bm{x})=k_g(\bm{x}^{(n)},\bm{x})$, respectively. 
By exploiting this structure, it is possible 
to recover the posterior GPs for~$f(\cdot)$ and~$g(\cdot)$ from the regression, and derive 
probabilistic uniform error bounds that
depend on the posterior standard deviation, as shown in the following theorem:\looseness=-1
\begin{lemma}
	\label{th:gpEbound}
	Consider a GP with composite kernel given by~\eqref{eq:compkern}, 
	a training data set~$\mathbb{D}$ and functions~$f(\cdot)$,~$g(\cdot)$,~$\hat{f}(\cdot)$
	and~$\hat{g}(\cdot)$ satisfying Assumptions~\ref{ass:data}, \ref{ass:prior} and \ref{ass:gp}.
	For any~$\delta\in(0,1)$ and~$\tau\in\mathbb{R}_+$, it holds that
	\begin{align}
	\label{eq:fbound}
		P\!\left(\! |f\!(\bm{x})\!-\!\mu_{\!f}\!(\bm{x})|\!\leq\! \sqrt{\!\beta\!(\tau)}\sigma_{\!f}\!(\bm{x})\!+\!\gamma_{\!f}\!(\tau)\!,\!\forall\bm{x}\!\in\!\mathbb{X} \!\right)&\!\geq\! 1\!-\!\delta\!
	\end{align}\vspace*{-0.4cm}%
	\begin{align}		
	\!P\!\left(\! |g(\bm{x})\!-\!\mu_{\!g}\!(\bm{x})|\!\leq\! \sqrt{\!\beta\!(\tau)}\sigma_{\!g}\!(\bm{x})\!+\!\gamma_{\!g}\!(\tau)\!,\!\forall\bm{x}\!\in\!\mathbb{X} \!\right)&\!\geq\! 1\!-\!\delta\!,\!
		\label{eq:gbound}
	\end{align}
	with mean and variance components
	\begin{align}
	\label{eq:muf}
		\mu_f(\bm{x})&\!=\!\hat{f}(\bm{x})\!+\!\bm{k}_f^T(\bm{x})(\bm{K}\!+\!s_{\mathrm{on}}^2\bm{I}_N)^{-1}\bm{y}\\
		\label{eq:mug}
		\mu_g(\bm{x})&\!=\!\hat{g}(\bm{x})\!+\!\bm{k}_g^T(\bm{x})\bm{U}(\bm{K}\!+\!s_{\mathrm{on}}^2\bm{I}_N)^{-1}\bm{y}\\
		\label{eq:sigf}
		\sigma_f^2(\bm{x})&\!=\!k_f(\bm{x},\bm{x})\!-\!\bm{k}_f^T(\bm{x})(\bm{K}\!+\!s_{\mathrm{on}}^2\bm{I}_N)^{-1}\bm{k}_f(\bm{x})\\
		\label{eq:sigg}
		\sigma_g^2(\bm{x})&\!=\!k_g(\bm{x},\bm{x})\!-\!\bm{k}_g^T(\bm{x})\bm{U}(\bm{K}\!+\!s_{\mathrm{on}}^2\bm{I}_N)^{-1}
		\bm{U}\bm{k}_g(\bm{x}),
	\end{align}
	and parameters 
	\begin{align}
	\beta(\tau)&=2d_x\log\left(1+\frac{r_0}{\tau}\right)-2\log(\delta)\\
	\gamma_f(\tau)&=(L_{\mu_f}+L_f)\tau+\sqrt{\beta(\tau)L_{\sigma_f^2}\tau}\\
	\gamma_g(\tau)&=(L_{\mu_g}+L_g)\tau+\sqrt{\beta(\tau)L_{\sigma_g^2}\tau}.
	\end{align}
	Here,~$L_{\mu_f}$,~$L_{\mu_g}$,~$L_{\sigma_f^2}$ and~$L_{\sigma_g^2}$ are the 
	Lipschitz constants of the mean and variance components, respectively, 
	and~$r_0=\max_{\bm{x},\bm{x}'\in\mathbb{X}}\|\bm{x}-\bm{x}'\|$ denotes 
	the maximum diameter of~$\mathbb{X}$.\looseness=-1
\end{lemma}
\begin{proof}
	It has been shown in \cite{Duvenaud2014} that the independent components~$f(\cdot)$
	and~$g(\cdot)u$ can be inferred by
	\begin{align}
		f(\bm{x})|\mathbb{D}&\sim\mathcal{N}(\mu_f(\bm{x}),\sigma_f^2(\bm{x}))\\
		g(\bm{x})u|\mathbb{D}&\sim\mathcal{N}(\mu_g(\bm{x})u,u^2\sigma_g^2(\bm{x}))
	\end{align}
	due to the structure of the kernel. Since $g(\bm{x})u$ depends linearly on $u$, we can
	extract
	\begin{align}
		g(\bm{x})|\mathbb{D}&\sim\mathcal{N}(\mu_g(\bm{x}),\sigma_g^2(\bm{x})).
	\end{align}
	Based on these identities, it is straightforward to adapt \cite[Theorem 3.1]{Lederer2019}
	to obtain the uniform error bounds \eqref{eq:fbound},~\eqref{eq:gbound}.\looseness=-1
\end{proof}

It is well known from scattered data approximation \cite{Wendland2005} that training data which covers $\mathbb{X}$ well
typically leads to small posterior variances, thereby implying that the learned 
model has a high accuracy. \cref{th:gpEbound} also exhibits this behavior, even 
though it additionally depends on the constants~$\gamma_f(\tau)$ and~$\gamma_g(\tau)$. Since these 
constants can be made arbitrarily small by reducing the value of $\tau$, their effect is usually negligible. 
In fact, bounds \eqref{eq:fbound} and \eqref{eq:gbound} can be shown to converge to $0$ under 
weak assumptions~\cite{Lederer2019}.\looseness=-1

\setlength{\abovedisplayskip}{5.2pt}
\setlength{\belowdisplayskip}{5.2pt}

\section{QUALITY ASSESSMENT OF TRAINING DATA FOR LEARNING-BASED CONTROL}
\label{sec:Control}

\subsection{Lyapunov-based Quality Assessment}
\label{subsec: data quality}

Although GPs are frequently used in control design, 
the relationship between training data and the resulting performance of a control law $u=\pi(\bm{x})$
has barely been analyzed. Therefore, there is typically little
insight on where training samples should be placed to achieve the highest improvement in
control performance. In the sequel, we measure the control performance using a Lyapunov 
function $V:\mathbb{R}^{d_x}\rightarrow\mathbb{R}_{+,0}$.
Therefore, we investigate
the time derivative of the Lyapunov function for systems defined in \eqref{eq:sys} given by\looseness=-1
\begin{align}
	&\!\dot{V}(\bm{e})=\nonumber\\
	&\!\sum\limits_{i=1}^{d_x-1}\!\frac{\partial V(\bm{e})}{\partial e_i}e_{i+1}
	\!\!+\!\frac{\partial V(\bm{e})}{\partial e_{d_x}}\!\left(\!\!f(\bm{x})\!+\!g(\bm{x})\pi(\bm{x})\!-\!\frac{\mathrm{d}^{d_x}x_{d}}{\mathrm{d}t^{d_x}}
	\!\!\right)\!.\!
\end{align}
By employing a Gaussian process model, as presented in
\cref{sec:GPR}, and exploiting 
the uniform error bound for GPs from~\cref{th:gpEbound}, 
we can bound this derivative by\looseness=-1
\begin{align}
	\dot{V}(\bm{e})&\leq \dot{V}_{\mathrm{nom}}(\bm{e})+\dot{V}_{\sigma_f}(\bm{e})+\dot{V}_{\sigma_g}(\bm{e}),
\end{align}
where the nominal component of the Lyapunov derivative is computed based on the GP mean function as
\begin{align}
	&\!\dot{V}_{\mathrm{nom}}(\bm{e})= \nonumber\\
	&\!\sum\limits_{i=1}^{d_x\!-\!1}\!\!\frac{\partial V(\bm{e})}{\partial e_i}e_{i+1}
	\!\!+\!\frac{\partial V(\bm{e})}{\partial e_{d_x}}\!\!\left(\!\!\mu_f\!(\bm{x})\!+\!\mu_g\!(\bm{x})\pi(\bm{x})\!-\!\frac{\mathrm{d}^{d_x}x_{d}}{\mathrm{d}t^{d_x}}\!\!\right)\!\!.\!
\end{align}
The uncertain component of $\dot{V}(\cdot)$ is separated 
into components for the uncertainty about $f(\cdot)$ and 
$g(\cdot)$, given by
\begin{align}
\label{eq:dV_sf}
	\dot{V}_{\sigma_f}(\bm{e})&=\left|\frac{\partial V(\bm{e})}{\partial e_{d_x}}\right|(\sqrt{\beta(\tau)}\sigma_f(\bm{x})+\gamma_f(\tau))\\
	\dot{V}_{\sigma_g}(\bm{e})&=\left|\frac{\partial V(\bm{e})}{\partial e_{d_x}}\right|(\sqrt{\beta(\tau)}\sigma_g(\bm{x})+\gamma_g(\tau))|\pi(\bm{x})|.
	\label{eq:dV_sg}
\end{align}
The nominal component of the Lyapunov function derivative does not depend on the 
uncertainty. Hence, it does not provide insight into the relationship between 
training data and control performance. In contrast, 
$\dot{V}_{\sigma_f}(\cdot)$ and $\dot{V}_{\sigma_g}(\cdot)$ directly depend on the 
GP posterior standard deviations, and thereby on the training data density. In order to measure this 
density in a flexible way, we introduce the~$M$-fill distance, inspired by classical concepts 
from scattered data approximation \cite{Wendland2005}.\looseness=-1

\begin{definition}
	The $M$-fill distance $\phi_{\bar{u}\!,\ubar{u}\!,M}(\bm{x})$ at a point $\bm{x}\in\mathbb{X}$ 
	is defined as the minimum radius~$\varphi$ of a ball with center $\bm{x}$, such that the ball
	contains~$M$ training samples~$\bm{z}^{(\!n\!)}$
	with control inputs~$\ubar{u}\!\leq\! |u^{(\!n\!)}|\!\leq\! \bar{u}$ for 
	some~$\ubar{u},\bar{u}\in\mathbb{R}_{+,0}$, i.e.,\looseness=-1
	\begin{subequations}
		\begin{align}
		\!\phi_{\bar{u}\!,\ubar{u}\!,M}(\bm{x})\!&=\!\min\limits_{\varphi\in\mathbb{R}_{+,0}} \varphi\\
		\!\text{s.t. }\!  &\!\left|\! \left\{\! \bm{z}\!^{(\!n\!)}\!\!\in\!\mathbb{D}\!\!:\! \|\bm{x}\!-\!\bm{x}\!^{(\!n\!)}\!\|\!\leq\! \varphi 
		\!\wedge\! \ubar{u}\!\leq\! |u\!^{(\!n\!)}\!|\!\leq\! \bar{u} \!\right\}\! \right|\!\geq\!M\!.
		\end{align}
	\end{subequations}\tiny{~}
\end{definition}
The~$M$-fill distance~$\phi_{\bar{u},\ubar{u},M}(\bm{x})$ measures the distance to 
the~$M$ closest training samples, where the parameter $M$ is used to 
adapt~$\phi_{\bar{u},\ubar{u},M}(\bm{x})$  to the total number of training samples $N$. Intuitively, 
one should choose $M\ll N$, such that only training 
points in the proximity of $\bm{x}$ are relevant for~$\phi_{\bar{u},\ubar{u},M}(\bm{x})$, thereby making it 
a local measure of the data density. A low~$M$-fill 
distance implies a high data density and indeed, 
upper bounds for~$\phi_{\bar{u},\ubar{u},M}(\bm{x})$ can be derived to 
guarantee a desired behavior $\xi_f(\cdot)$ and $\xi_g(\cdot)$ for~$\dot{V}_{\sigma_f}(\cdot)$ and~$\dot{V}_{\sigma_g}(\cdot)$,
respectively.\looseness=-1
\begin{theorem}
	\label{cor:datf}
	If the~$M$-fill
	distance~$\phi_{\bar{u}_f,0,M}\!(\cdot)$ satisfies\looseness=-1
	\begin{align}
	\phi_{\bar{u}_f,0,M}^2(\bm{x})\leq \bar{\phi}_{f}^2(\bm{x})+\theta_f
	\end{align}	
	for all~$\bm{x}\in\mathbb{X}$, where
	\begin{align}
	\label{eq:phif}
	\bar{\phi}_{f}^2(\bm{x})\!&=
	-l_f^2\!\log\! \left(\!\! 1\!-\!
	 \frac{\left(\xi_f(\bm{e})-\gamma_f(\tau)\left|\frac{\partial V(\bm{e})}{\partial e_{d_x}}\right|\right)^2}{\beta(\tau)s_f^2\left|\frac{\partial V(\bm{e})}{\partial e_{d_x}}\right|^2}
	 \!\right)\\
	\theta_f&=-l_f^2\log\left(\!\!\frac{Ms_{f}^2\!+\!Ms_{g}^2\bar{u}_f^2\!+\!s_{\mathrm{on}}^2}{Ms_{f}^2} \!\!\right)
	\label{eq:thetaf}
	\end{align}
	for any~$\bar{u}_f\!\in\!\mathbb{R}_{+,0}$,
	$M\!\in\!\mathbb{N}$, and~$\xi_f(\bm{e})\!>\!\gamma_f(\tau)|\partial V(\bm{e})/\partial e_{d_x}|$, 
	
	\scalebox{.001}{$~$}\vspace*{-0.14cm}\\
	\noindent
	then,~$\dot{V}_{\sigma_f}(\bm{e})\!\leq\!\xi_f(\bm{e})$,~$\forall \bm{x}\!\in\!\mathbb{X}$.\looseness=-1
\end{theorem}
\begin{proof}
	In order to prove this lemma, we have to bound the posterior standard deviation~$\sigma_f(\bm{x})$.
	Following the  approach introduced in \cite{Lederer2019a}, this is achieved by considering only 
	$M$ training samples~$\bm{z}^{(n)}$ within distance~$\phi_{\bar{u},\ubar{u},M}(\bm{x})$ to~$\bm{x}$ 
	such that the posterior variance is bounded by\footnote{
		We do not state the dependency on~$\bm{x}$ explicitly if it arises from the restriction 
		of the considered training samples for notational simplicity.\looseness=-1}
	\begin{align}
	\sigma_f^2(\bm{x})\leq s_{f}^2-\frac{\|\bm{k}_{f,M}(\bm{x})\|^2}{\lambda_{\max}(\bm{K}_M)+s_{\mathrm{on}}^2},
	\end{align}
	where~$\bm{k}_{f,M}(\bm{x})$ and~$\bm{K}_M$ denote the covariance vector and matrix based on these~$M$ 
	samples and~$\lambda_{\max}(\bm{K}_M)$ denotes the maximum eigenvalue. 
	Application of the Gershgorin theorem allows us to bound the maximum 
	eigenvalue by\looseness=-1
	\begin{align}
	\lambda_{\max}(\bm{K}_M)\leq M(s_{f}^2+\bar{u}_f^2s_{g}^2).
	\end{align}
	Since we have~$\|\bm{k}_{f\!,M}\!(\bm{x})\|^2\!\geq\!\! M s_{f}^4\!\exp(\!-\phi_{\bar{u}_f\!,0,M}^2\!(\bm{x})/l_f^2)$,
	we obtain the posterior variance bound
	\begin{align}
	\sigma_f^2(\bm{x})&\leq  s_{f}^2-\frac{s_{f}^4\exp\left(-\frac{\phi_{\bar{u}_f\!,0,M}^2(\bm{x}) }{l_f^2}\right)}{s_{f}^2+\bar{u}_f^2s_{g}^2+\frac{s_{\mathrm{on}}^2}{M}}.
	\end{align}
	Substituting this expression into \eqref{eq:dV_sf} and solving for~$\phi_{\bar{u}_f\!,0,M}^2(\bm{x})$ yields the desired result.	
\end{proof}
\begin{corollary}
	\label{th:datg}
	If the $M$-fill distance~$\phi_{\bar{u}_g,\ubar{u}_g,M}(\cdot)$ satisfies
	\begin{align}
	\phi^2_{\bar{u}_g,\ubar{u}_g,M}(\bm{x})\leq \bar{\phi}_{g}^2(\bm{x})+\theta_g,
	\end{align}
	for all~$\bm{x}\in\mathbb{X}$, where
	\begin{align}
	\label{eq:phig}
	\!\bar{\phi}_g^2(\bm{x})\!&=-l_g^2\log\!\!\left(\!\! 1\!-\! 
	\frac{\!\left(\!\xi_g(\bm{e})\!-\!\gamma_g(\tau)|\pi(\bm{x})|\!\left|\!\frac{\partial V(\bm{e})}{\partial e_{d_x}}\!\right|\right)^{\!\!2}\!}{\beta(\tau)s_g^2|\pi(\bm{x})|^2\!\left|\!\frac{\partial V(\bm{e})}{\partial e_{d_x}}\!\right|^2}
	\right)\!\\
	\theta_g&=-l_g^2\log\left(\!\frac{Ms_{f}^2\!+\!M\bar{u}_g^2s_{g}^2\!+\!s_{\mathrm{on}}^2}{Ms_{g}^2\ubar{u}_g^2} \!\right)\!
	\label{eq:thetag}
	\end{align}
	for any~$\ubar{u}_g,\bar{u}_g\!\!\in\!\!\mathbb{R}_{+,0}$,~$\ubar{u}\!\!\leq\!\!\bar{u}$,
	$M\!\!\in\!\!\mathbb{N}$, and~$\xi_g(\bm{e})\!\!>\!\!\gamma_g(\tau)|\pi(\bm{x})| |\partial V(\bm{e})/\partial e_{d_x}|$, 
	then,~$\dot{V}_{\sigma_g}(\bm{e})\!\!\leq\!\xi_g(\bm{e})$,~$\forall \bm{x}\!\in\!\mathbb{X}$.\looseness=-1
\end{corollary}
\begin{proof}
	We can bound the posterior standard deviation~$\sigma_f(\bm{x})$ analogously to 
	\cref{cor:datf} since 
	\begin{align}
	\sigma_g^2(\bm{x})\leq s_{g}^2-\frac{\|\bm{k}_{g,M}(\bm{x})\bm{U}\|^2}{\lambda_{\max}(\bm{K}_M)+s_{\mathrm{on}}^2}.
	\end{align}
	The remainder of this proof follows directly from a straightforward
	adaptation of the proof of \cref{cor:datf}.
\end{proof}

\cref{cor:datf} and \cref{th:datg} allow to 
directly investigate if~$\dot{V}_f(\cdot)$ and~$\dot{V}_g(\cdot)$ 
satisfy a desired behavior by measuring the $M$-fill distance.
Therefore, they 
provide helpful insight on how the training data should be distributed. 
The quality of the data for learning the decoupling between $f(\cdot)$ and $g(\cdot)$ 
is measured by $\theta_f$ and $\theta_g$.
It is straightforward to see that
$\theta_g$ is close to zero if~$\bar{u}_g\approx\ubar{u}_g$ is large.  This is 
intuitive since the weight of~$g(\cdot)$ 
in the function~$h(\bm{z})=f(\bm{x})+g(\bm{x})u$ grows linearly with~$u$. Thus, large 
control inputs are beneficial for the identification of~$g(\cdot)$. In contrast, \cref{cor:datf} shows that small 
control inputs~$u\approx 0$ are advantageous for 
learning~$f(\cdot)$ since the control dependency of
$h(\cdot)$ disappears. 
Finally, large noise variance
requires higher data density to keep the term~$s_{\mathrm{on}}^2/M$ small in both bounds.\looseness=-1

In contrast to $\theta_f$ and $\theta_g$, the 
functions~$\bar{\phi}_{f}^2(\bm{x})$ 
and~$\bar{\phi}_{g}^2(\bm{x})$ express the dependency on the 
performance specification. Since $\gamma_f(\tau)$ and $\gamma_g(\tau)$ are 
usually negligible, the quotients in~\eqref{eq:phif} and \eqref{eq:phig} can be approximated by
the ratio between the squared performance specifications~$\xi_f^2(\cdot)$, $\xi_g^2(\cdot)$ and 
the prior bounds $\beta(\tau)s_f^2|\partial
V(\bm{e})/\partial e_{d_x}|^2$, $\beta(\tau)s_g^2|\pi(\bm{x})|^2|\partial
V(\bm{e})/\partial e_{d_x}|^2$, respectively. Hence, small performance specifications,  e.g., 
$\xi_f(\cdot)\ll \beta(\tau)s_f^2|\partial V(\bm{e})/\partial e_{d_x}|^2$, cause small values 
of~$\bar{\phi}_f(\cdot)$, $\bar{\phi}_g(\cdot)$, which in turn indicate the necessity of high data density.
Since the performance specifications define the allowed increase of $\dot{V}_{\sigma_f}(\cdot)$, $\dot{V}_{\sigma_g}(\cdot)$, 
it is natural to define them as a fraction of the nominal 
derivative, e.g., \looseness=-1
\begin{align}
\label{eq:Vnom frac}
\xi_f\!(\cdot)\!=\!\chi_f\dot{V}_{\mathrm{nom}}\!(\cdot),~~ \xi_g\!(\cdot)\!=\!\chi_g\dot{V}_{\mathrm{nom}}\!(\cdot),~~ \chi_f\!,\!\chi_g\!\in\! \mathbb{R}_+,
\end{align}
such that stability is guaranteed for $\chi_f\!+\!\chi_g\!<\!1$. Due to this intuitive
interpretation of the~$M$-fill distance, we propose to use it as the basis for measures
of the quality of the training data distribution.\looseness=-1
\def\xstrutf{\rule{0pt}{1.75ex}}
\def\xstrutg{\rule{0pt}{1.6ex}}
\begin{definition}
The $\rho\!_{\xstrutf f}$- and $\rho\!_{\xstrutg g}$-gaps are defined as
\begin{align}
\rho\!_{\xstrutf f}(\bm{x})&=\phi_{\bar{u}_f,0,M}(\bm{x})-\bar{\phi}_f^2(\bm{x})\\
\rho\!_{\xstrutg g}(\bm{x})&=\phi_{\bar{u}_g,\ubar{u}_g,M}(\bm{x})-\bar{\phi}_g^2(\bm{x}).
\end{align}
\end{definition}

The $\rho$-gaps measure the difference 
between required $M$-fill distances~$\bar{\phi}_f^2(\cdot)$, $\bar{\phi}_g^2(\cdot)$, which express the 
dependency of the data density on the desired bounds~$\xi_f(\cdot)$, $\xi_g(\cdot)$ for~$\dot{V}_{\sigma_f}(\cdot)$, $\dot{V}_{\sigma_g}(\cdot)$, 
and the actual $M$-fill distances~$\phi_{\bar{u}_f,0,M}(\cdot)$, $\phi_{\bar{u}_g,\ubar{u}_g,M}(\cdot)$, which are independent of the control problem.
Therefore, searching for points~$\bm{x}\in\mathbb{X}$
with high $\rho\!_{\xstrutf f}$- and $\rho\!_{\xstrutg g}$-gaps yields regions, where the distance between 
training samples is too large to satisfy the bounds $\xi_f(\cdot)$, $\xi_g(\cdot)$. This can be exploited, e.g., in 
active learning and training data selection, to find points with high $\rho$-gap and add them to the training
set in order to reduce the violation 
of the performance specification. Note that $\theta_f$ and $\theta_g$ are not included in the definitions 
of $\rho\!_{\xstrutf f}(\cdot)$ and $\rho\!_{\xstrutg g}(\cdot)$, respectively, since they are independent of the state, and thereby do 
not provide useful information about the distribution of the data.\looseness=-1

\subsection{Ultimate Error Bound for Feedback Linearization}
\label{subsec:ult bound}

In order to demonstrate the intuitive relationship between the tracking error 
and the proposed $\rho$-gaps, we extend existing stability results for feedback linearization 
with GP models. According to \cite{Umlauft2020}, 
we define the filtered state 
$r=[\bm{\lambda}^T\quad 1]\bm{e}$ with Hurwitz coefficients 
$\bm{\lambda}\in\mathbb{R}^{d_x-1}$. Based on the GP model from \cref{sec:GPR}, 
we define the approximately linearizing control law\looseness=-1
\begin{align}
\label{eq:FeLi}
\pi(\bm{x})=\mu_g^{-1}(\bm{x})(\nu-\mu_f(\bm{x})),
\end{align}
where~$\nu$ is the input to the linearized system, and~$\mu_f(\cdot)$ and~$\mu_g(\cdot)$ are
defined in \eqref{eq:muf} and \eqref{eq:mug}, 
respectively. The existence of the reciprocal value of~$\mu_g(\cdot)$ can be ensured 
by a suitable choice of hyperparameters and prior mean~$\hat{g}(\cdot)$ as shown 
in \cite[Proposition 1]{Umlauft2020}. In order to achieve perfect tracking behavior
for known functions~$f(\cdot)$ and~$g(\cdot)$, we define the input to the linearized 
system as\looseness=-1
\begin{align}
\label{eq:nu}
\nu = -k_c r-\bm{\lambda}^T\bm{e}_{2:d_x}+\frac{\mathrm{d}^{d_x}x_d}{\mathrm{d}t^{d_x}},
\end{align}
where~$k_c\in\mathbb{R}_+$ denotes the control gain and~$\bm{e}_{2:d_x}=[e_2\ \cdots\ e_{d_x}]^T$. Following the 
approach in \cref{subsec: data quality}, 
we can determine the ultimately bounded set for the 
system in~\eqref{eq:sys} controlled by \eqref{eq:FeLi} as shown in the following theorem.\looseness=-1

\begin{theorem}
	\label{th:stab}
	Consider a system \eqref{eq:sys}, a prior model~$\hat{f}(\cdot)$, $\hat{g}(\cdot)$ 
	and training data~$\mathbb{D}$ satisfying Assumptions~\ref{ass:data}-\ref{ass:gp}.
	If\looseness=-1
	\begin{align}
	\label{eq:alpha}
		\alpha(\bm{x})=\frac{\sqrt{\beta(\tau)}\sigma_g(\bm{x})+\gamma_g(\tau)}{\mu_g(\bm{x})}<\eta \quad\forall \bm{x}\in\mathbb{X}
	\end{align}
	with
	\begin{align}
		\!\eta \!=\!\min\!\left\{\! \frac{k_c\lambda_2}{k_c\lambda_2\!+\!\lambda_1}\!,\ldots,\frac{k_c\lambda_{d_x\!-\!1}}{k_c\lambda_{d_x\!-\!1}\!+\!\lambda_{d_x\!-\!2}}\!,\frac{k_c}{k_c\!+\!\lambda_{d_x\!-\!1}} \!\right\}\!,\!
	\end{align}
	then there exists a control gain~$k_c$, such that the tracking error $\bm{e}$ obtained with the 
	feedback linearizing controller~\eqref{eq:FeLi} with input to the linearized system \eqref{eq:nu}
	converges, with probability of at least~$1-2\delta$, to the ultimately bounded set
	\begin{align}
	\label{eq:Bset imp}
	\!\mathbb{B}=\left\{\! \sqrt{\!\bm{e}^T\!\bm{\Lambda}\bm{e}}\leq\! \frac{\sqrt{\!\beta(\tau)}\sigma_{f}(\bm{x})\!+\!\gamma_{f}(\tau)
		\!+\!\alpha(\bm{x}) c(\bm{x}) }{\tilde{k}_c(\bm{x}) }\!\right\},
	\end{align}
	where
	\begin{align}
		\tilde{k}_c(\bm{x})=k_c \!\left( 1\!-\!\frac{\alpha(\bm{x})}{\eta } \right),&&
		c(\bm{x})=\left|\frac{\mathrm{d}^{d_x}x_{d}}{\mathrm{d}t^{d_x}}\!-\!\mu_{f}(\bm{x})\right|.
	\end{align}\tiny~
\end{theorem}
\begin{proof}
	It has been shown in \cite{Umlauft2020} that there exists a gain~$k_c$, such 
	that the closed-loop system is ultimately bounded to the compact set $\mathbb{X}$.
	Hence, we can restrict our analysis to this set. 
	Consider the Lyapunov function~$V(r)=\frac{1}{2}r^2$, which allows us
	to analyze the stability of the tracking error since~$\bm{\lambda}$ is Hurwitz 
	\cite{Yesildirek1995}. The derivative of this Lyapunov function is decomposed into
	\begin{align}
	\dot{V}(r)\leq-k_cr^2+\dot{V}_{\sigma_f}(r)+\dot{V}_{\sigma_g}(r),
	\end{align}
	where 
	\begin{align}
	\dot{V}_{\sigma_f}(r)&=|r|(\sqrt{\beta(\tau)}\sigma_f(\bm{x})+\gamma_f(\tau))\\
	\dot{V}_{\sigma_g}(r)&=|r|(\sqrt{\beta(\tau)}\sigma_g(\bm{x})+\gamma_g(\tau))|\pi(\bm{x})|.
	\end{align}
	 We can separate 
	the bound of the 
	control input into feedback and feedforward components
	\begin{align}
		|\mu_g(\bm{x})||\pi(\bm{x}))|\!\leq\!
		|k_c r+\bm{\lambda}^T\bm{e}_{2:d_x}|\!+\! \left|\frac{\mathrm{d}^{d_x}x_d}{\mathrm{d}t^{d_x}}\!-\!\mu_f(\bm{x})\right|.
	\end{align}
	The feedback component can be bounded 
	by
	\begin{align}
		|k_c r+\bm{\lambda}^T\bm{e}_{2:d_x}|\leq \frac{k_c}{\eta }|r|,
	\end{align}
	while the feedforward component is a bounded state dependent function \cite{Umlauft2020}. 
	Hence, we obtain the bound
	\begin{align}
	\!\dot{V}(r)\!\leq\! 
	&-\!\tilde{k}_c(\bm{x})r^2\!\!+\!|r|\!\bigg(\!\!\!\sqrt{\!\beta(\tau)}\sigma_{\!f}\!(\bm{x})\!+\!\gamma_{\!f}\!(\tau)\!+\!\alpha(\bm{x}) c(\bm{x}) \!\!\bigg)\!
	\end{align}
	due to the definition of~$\alpha(\bm{x})$ and $\tilde{k}_c(\bm{x})$. The quotient~$\alpha(\bm{x})/\eta$ is 
	smaller than one by assumption, such that $\tilde{k}_c(\bm{x})>0$ and the Lyapunov function derivative 
	becomes negative for all\linebreak
	$r>(\sqrt{\!\beta(\tau)}\sigma_{f}(\bm{x})\!+\!\gamma_{f}(\tau)
		\!+\!\alpha(\bm{x}) c(\bm{x}))/\tilde{k}_c(\bm{x})$. 
\end{proof}

While the condition on the sufficiently high control gain~$k_c$ is theoretically important to
ensure the global ultimate boundedness, it is practically sufficient to analyze the conditions of 
\cref{th:stab} on a set~$\mathbb{X}$ and choose~$k_c$ high enough to ensure that all points 
$\bm{x}_d+\bm{e}$ with~$\bm{e}\in\mathbb{B}$ are in the interior of~$\mathbb{X}$. Based on this
local analysis, an ultimate bound is obtained which holds for initial values~$\bm{x}(0)$ in a 
neighborhood of~$\bm{x}_d(0)$. Moreover, the condition on~$\alpha(\cdot)$ stems from the uncertainty
about~$g(\cdot)$ and ensures that its sign is robustly known under the posterior GP distribution. 
Due to this uncertainty, the effect of the feedback control on the tracking error bound is reduced,
resulting in a diminished effective control gain~$\tilde{k}_c(\cdot)$. Furthermore, the ultimate error 
bound can be made arbitrarily
small by increasing the effective gain~$\tilde{k}_c$, which is achieved by increasing the nominal 
gain~$k_c$ or reducing the uncertainty about~$g(\cdot)$. In fact, if the function~$g(\cdot)$ 
is known exactly, this effect disappears and we recover the ultimate bound which has been 
proposed in different forms in \cite{Umlauft2020, Lederer2019}.\looseness=-1

\section{NUMERICAL EVALUATION}
\label{sec:simulation}

\subsection{Experimental Setting}
\label{subsec:setting}
We investigate the learning-based controller \eqref{eq:FeLi} and the 
corresponding ultimate bound on the nonlinear system\looseness=-1
\begin{align}
	f(\bm{x})&=1-\sin(x_1)+\frac{1}{1+\exp(-x_2)}\\
	g(\bm{x})&=20\left(1+\frac{1}{2}\sin\left(\frac{x_2}{4}\right)\right).
\end{align}
This dynamical system is well-suited
for illustrating the proposed $\rho$-gaps. Both functions are slowly 
varying, such that the GP is capable of extrapolating to regions without data. Therefore, we 
do not need well distributed data and can investigate the effect of an increasing distance to training data
on the $\rho$-gaps. 
We express prior knowledge about this system using the approximate models
$\hat{f}(\bm{x})=0$,~$\hat{g}(\bm{x})=20$. Moreover, we define the reference trajectory
$x_d(t)=2\sin(t)$ and generate~$N=1000$ training samples by applying a high gain 
feedback linearizing controller based on the approximate models~$\hat{f}(\cdot)$, 
$\hat{g}(\cdot)$ to track the reference trajectory. We add zero mean Gaussian 
noise with standard deviation~$s_{\mathrm{on}}=0.5$ to the observed accelerations
and train a GP using log-likelihood maximization \cite{Rasmussen2006}. 
We approximate the Lipschitz constants~$L_{\mu}$ and~$L_{\sigma^2}$ of the 
resulting GP numerically and bound the Lipschitz constant of~$f(\cdot)$
and~$g(\cdot)$ by using twice the nominal value. 

In the numerical experiment, we simulate the system for~$T\!=\!30$ starting at~$\bm{x}(0)\!=\!\bm{0}$ 
with a control gain of~$k_c\!=\!40$ and~$\lambda\!=\!1$. The constants~$\beta(\tau)$,~$\gamma_f(\tau)$ 
and~$\gamma_g(\tau)$ are computed using~$\tau\!=\!10^{-4}$,~$\delta\!=\!10^{-2}$ and the 
conditions for the ultimately
bounded set \eqref{eq:Bset imp} are investigated on the set 
$\mathbb{X}\!\!=\!\{\bm{x}\!\in\!\mathbb{R}^2\!\!:\!\|\bm{x}\|\!\leq\! 2.5\}$.
For the definition of the performance specifications, we use \eqref{eq:Vnom frac} with 
$\chi_f\!=\!\chi_g\!=\!0.25$, such that their satisfaction ensures stability. We simplify the $\rho$-gaps analogously to the 
proof of \cref{th:stab} and measure the 
density of informative points for the identification of~$g(\cdot)$ by defining 
$\bar{u}_g=\max_{n\leq N} |u^{(n)}|$. Moreover, we define~$\ubar{u}_g$ such that~$90\%$ of the control 
inputs~$|u^{(n)}|$ are smaller than~$\ubar{u}_g$, and choose $M=1$. Similarly, we define~$\bar{u}_f$ 
such that~$90\%$ of the control inputs~$|u^{(n)}|$ are larger than~$\bar{u}_f$.\looseness=-1

\subsection{Results}
\label{subsec:results}

The evolution of the observed tracking error~$\sqrt{\bm{e}^T\bm{\Lambda}\bm{e}}$ 
and the maximum extension of the
ultimately bounded set~$\bar{e}_{\mathbb{B}}$ are depicted in \cref{fig:T_e_traj}. 
It can be clearly observed that the tracking error indeed satisfies the ultimate bound for
the given confidence level~$\delta=10^{-2}$ after a brief convergence period.
The curves for the ultimate 
error bound~$\bar{e}_{\mathbb{B}}$ and the tracking error $\sqrt{\bm{e}^T\!\bm{\Lambda}\bm{e}}$ 
exhibit a similar behavior with minima and maxima occurring 
at almost identical times. 

\begin{figure}[t]
	\centering 
	\vspace{0.2cm}
	\includegraphics{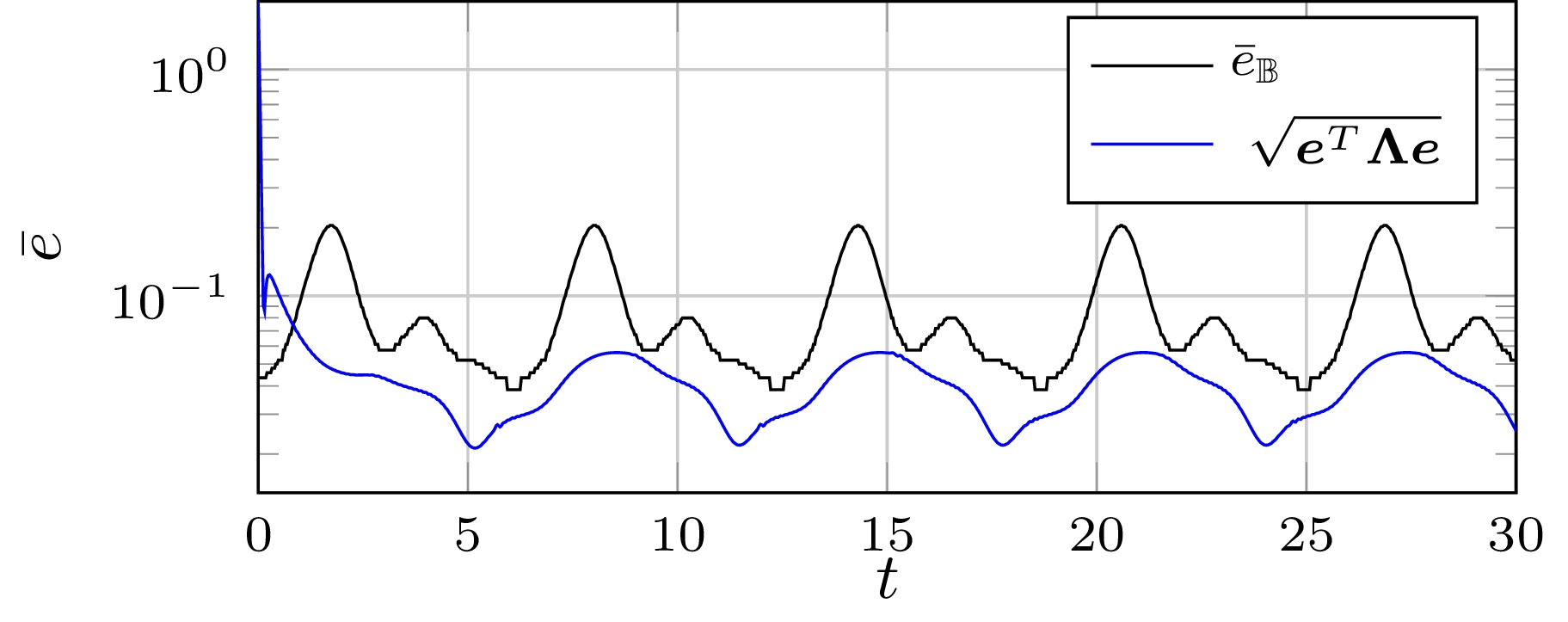}
	
	\caption{Radius of the ultimately bounded set (black) and tracking error 
		observed in simulation (blue). 
	}
	\label{fig:T_e_traj}
\end{figure}

Snapshots of the state trajectory as well as the corresponding
ultimate bound are depicted in \cref{fig:LyapDecrTe_Traj}. On the left hand side
the maximum of the ultimate bound occurs at the maximum of~$\rho\!_{\xstrutf f}(\cdot)$
along the reference trajectory. Taking a closer look at the training data, it can 
be seen that the training samples exhibit large control inputs~$u^{(n)}$ in this 
area, which leads to high uncertainty in the identification of~$f(\cdot)$. Therefore, 
there is a lack of training data, as indicated by~$\rho\!_{\xstrutf f}(\cdot)$. In contrast, 
the right hand side of \cref{fig:LyapDecrTe_Traj} shows the~$\rho\!_{\xstrutg g}$-gap, 
which is minimal at the minimum of the ultimate bound~$\bar{e}_{\mathbb{B}}$. This minimum
of~$\rho\!_{\xstrutg g}(\cdot)$ is a consequence of the choice of~$g(\cdot)$, which is increasing with respect to~$x_2$ 
in the considered state space~$\mathbb{X}$. Therefore, model uncertainties have a slightly weaker 
effect in the upper half plane than in the lower half plane.
Furthermore,
it can be clearly observed that~$\rho\!_{\xstrutg g}(\cdot)$ exhibits a maximum along the trajectory
at~$[-2\ 0]^T$. In fact, this lack of training data with high control inputs~$u^{(n)}$ 
causes the second local maximum in the ultimate bound at~$2\pi m+4$,~$m\in\mathbb{N}$. Finally, both~$\rho$-gaps show different behavior with increasing distance to the reference trajectory: Due to the increasing nominal derivative $\dot{V}_{\mathrm{nom}}(\cdot)$, $\rho\!_{\xstrutf f}(\cdot)$ is strongly decreasing, while 
a similar growth of the control input $\pi(\bm{x})$ compensates this effect, causing a growing $\rho\!_{\xstrutg g}(\cdot)$.\looseness=-1

\section{CONCLUSION}
\label{sec:conclusion}
This paper introduces a quality measure for training sets in data-driven 
control. We establish a relationship between the distribution of the training 
data and the ultimate bound of the tracking error for a Gaussian process-based 
control law. In contrast to state-of-the-art information-theoretical 
quantities, our measure allows us to determine the most useful data points for 
control. In future work, this can be used to design exploration algorithms that 
collect data such that the control performance of the control law is maximized. 
\looseness=-1

\begin{figure}
	\pgfplotsset{width=120\textwidth /100, compat = 1.13, 
		height =120\textwidth /100, grid= major, 
		legend cell align = left, ticklabel style = {font=\scriptsize},
		every axis label/.append style={font=\small},
		legend style = {font=\scriptsize},title style={yshift=-7pt, font = \small} }
		\center

		\begin{minipage}{0.48\columnwidth}
		\vspace{0.1cm}
		\includegraphics{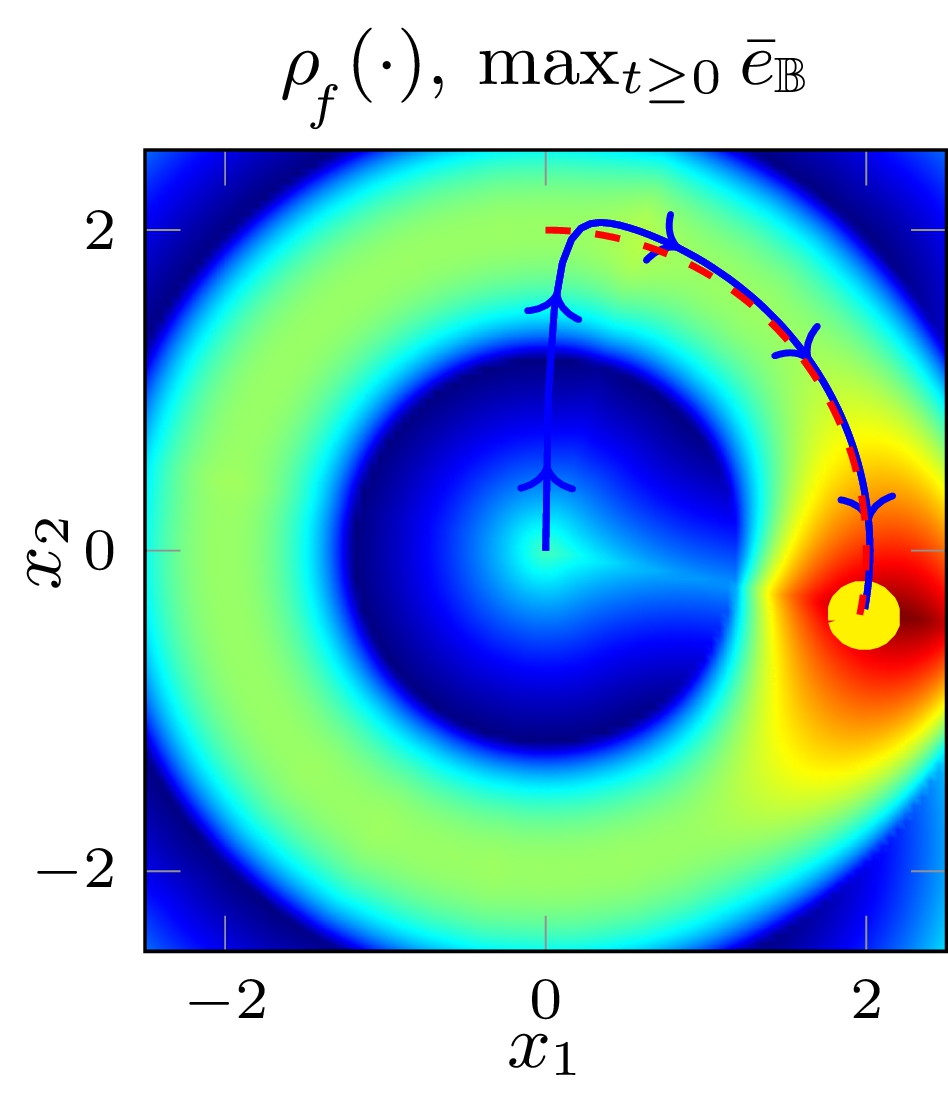}
		
		\end{minipage}
		\begin{minipage}{0.48\columnwidth}
		\vspace{0.1cm}
		\includegraphics{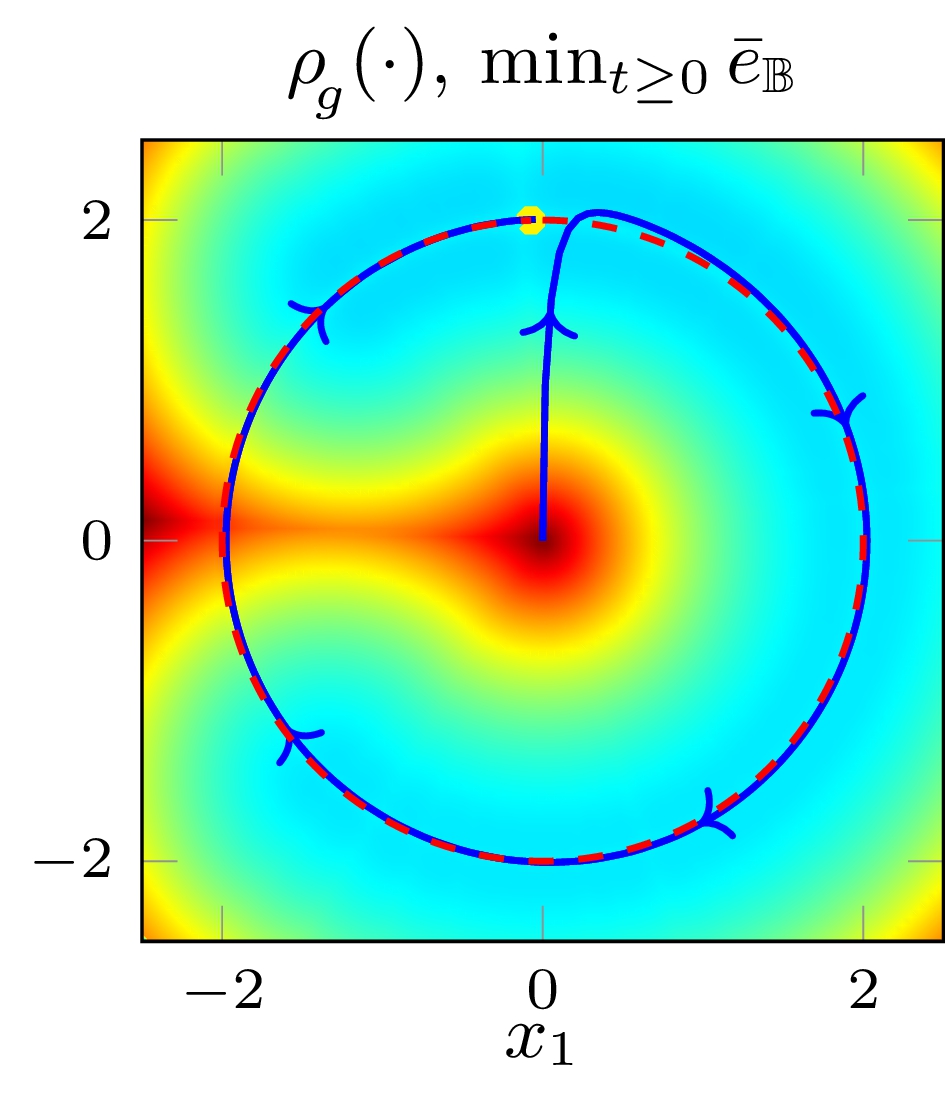}
		
		\end{minipage}
		
		\hspace{-1.5cm}
		\begin{minipage}{0.48\columnwidth}
		\includegraphics{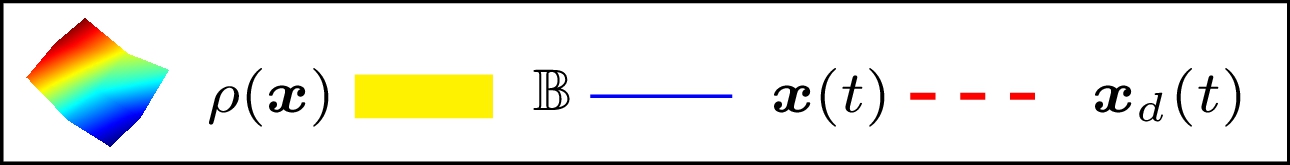}\\\tiny~\\
		\end{minipage}
	\caption{Snapshots of the state trajectory (blue) and reference trajectory (red)
		of the system controlled with~$N=1000$ 
		training samples obtained using a high gain controller. Large ultimate bounds are strongly related
		to large values of~$\rho\!_{\xstrutf f}(\cdot)$ and~$\rho\!_{\xstrutg g}(\cdot)$ (red background).\looseness=-1
		}
	\label{fig:LyapDecrTe_Traj}
\end{figure}

\bibliographystyle{IEEEtran}
\bibliography{IEEEabrv,myBib}

\end{document}